\documentclass[peerreview]{IEEEtran}
\usepackage{graphicx}

\usepackage[utf8]{inputenc}
\usepackage{amsmath,amssymb}
\usepackage{float}
\usepackage{epsfig}
\usepackage{graphics}
\usepackage{graphicx}
\usepackage{epstopdf}
\usepackage{subcaption}
\usepackage{algorithm}
\usepackage{algpseudocode}
\usepackage{pifont}
\usepackage[utf8]{inputenc}
\usepackage{amsmath,amssymb}
\usepackage{float}
\usepackage{epsfig}
\usepackage{graphics}
\usepackage{graphicx}
\usepackage{color}
\usepackage{latexsym}
\usepackage{epstopdf}
\usepackage{subcaption}
\usepackage{algorithm}
\usepackage{algpseudocode}
\usepackage{pifont}
\usepackage[utf8]{inputenc}
\usepackage[english]{babel}
\usepackage{amsthm}

\newtheorem{lemma}[]{Lemma}
\newtheorem{assumption}{Assumption}
\newtheorem{remark}{Remark}
 \newcommand\subparagraph{}
\usepackage{url}

\begin{document}
%
\title{Optimal Radio Access Technology Selection Algorithm for LTE-WiFi Network}
%
%
%
\author{Arghyadip Roy, Prasanna Chaporkar and Abhay Karandikar\\
Department
of Electrical Engineering\\ Indian Institute of Technology Bombay,
Mumbai, India, 400076\\
e-mail: {$\lbrace$arghyadip, chaporkar, karandi$\rbrace$}@ee.iitb.ac.in}

\maketitle

\begin{abstract}
A Heterogeneous Network (HetNet) comprises of multiple Radio Access Technologies (RATs) allowing a user to associate with a specific RAT and steer to other
RATs in a seamless manner. To cope up with the unprecedented growth of data traffic, mobile data can be offloaded to Wireless Fidelity (WiFi) in a Long 
Term Evolution (LTE) based HetNet.
In this paper, an optimal RAT selection problem is considered to maximize the total system throughput in an LTE-WiFi system with offload capability. 
Another formulation is also developed where maximizing the total system throughput is subject to a constraint
on the voice user blocking probability. It is proved that the optimal policies for the association and offloading of voice/data users contain 
threshold structures. Based on the threshold structures, we propose algorithms for the association and offloading of users in LTE-WiFi HetNet. 
Simulation results are presented to demonstrate the voice user blocking probability and the total system throughput performance of the proposed algorithms 
in comparison to another benchmark algorithm.   
\end{abstract}

\begin{IEEEkeywords}
User association, LTE-WiFi offloading, CMDP, Threshold policy.
\end{IEEEkeywords}

%

\section{Introduction}
To meet the ever-increasing Quality of Service (QoS) requirements of users, various Radio Access Technologies (RATs) have been standardized \cite{hetnet}. 
Each RAT has different characteristics regarding associated parameters like coverage and capacity. 
It has been predicted that by $2021$ monthly 
global mobile data traffic will exceed $49$ exabytes \cite{cisco}. This unprecedented growth in data traffic has become one of the serious challenges for 
cellular network operators. 
To address this issue, both from users' and network providers' point of view, it has become necessary that different RATs interwork with each other.
A wireless network where different RATs are present, and users can be associated and moved seamlessly from one RAT to another, is called a Heterogeneous 
Network (HetNet). In this paper, our aim is to determine the optimal RAT selection policy in a HetNet \footnote{The terminologies ``RAT selection'' 
and ``association'' has been used interchangeably throughout the paper.}. \par 
Due to the complementary characteristics of Third Generation Partnership Project (3GPP) Long Term Evolution (LTE) Base Stations (BSs) providing ubiquitous 
coverage and IEEE 802.11 \cite{ieee} based Wireless Local Area Network (WLAN) (also known as Wireless Fidelity (WiFi)) Access Points (APs) providing high 
bit rate capability in hot-spot areas, interworking between them \cite{interworking} offers an interesting solution. 
In areas where both LTE and WiFi coverage are present, a user can be associated with either of them. Moreover, data users can be steered from one RAT to 
another to achieve load balancing. This proposal, known as mobile data offloading, has been introduced in $3$GPP Release $12$ specifications 
\cite{interworking}. Since WiFi operates in unlicensed spectrum and most of the commercially available user equipments already have a dedicated WLAN 
interface, this proposal has become popular both with network operators and handset manufacturers.\par
For efficient utilization of both LTE and WiFi networks, it is necessary to take appropriate association and offloading decisions. RAT selection and 
offloading decisions can be made either at the user side or the network side. In user-initiated RAT selection schemes, there is no cooperation between 
LTE and WiFi networks, and users decide which RAT should be selected based on certain criteria. Since users individually take selfish RAT selection decisions
to maximize individual utility functions, this may not provide a globally optimum solution \cite{mikhail}-\cite{game}. To address this issue, a 
network-initiated RAT selection algorithm, which optimizes different network parameters, becomes necessary. \par
\begin{figure}[!h]
 \begin{center} 
\includegraphics[width=0.30\textwidth]{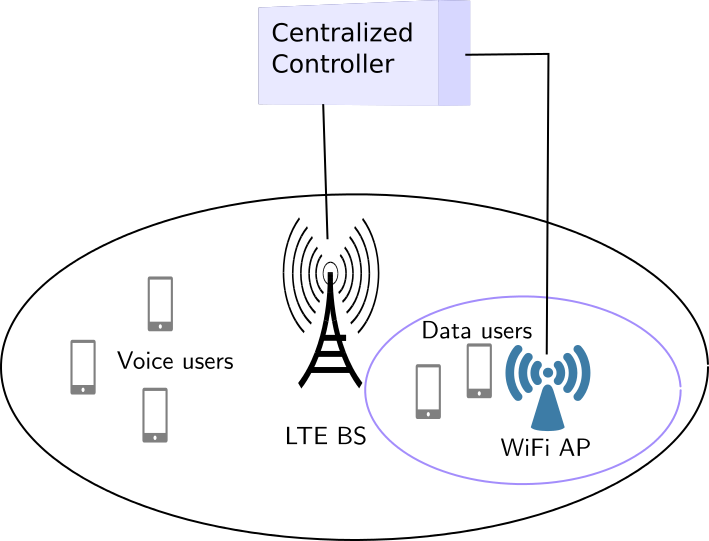}
\caption{LTE-WiFi heterogeneous network architecture.}
\label{fig:hetnet}
 \end{center}
\end{figure}
In this paper, we investigate an optimal association policy for an LTE-WiFi HetNet, as illustrated in Fig.\ref{fig:hetnet}. Network-initiated RAT selection 
and offloading decisions are taken by a centralized controller possessing an overall view of the network. We consider two types of users, viz., voice and data users, to be present inside the LTE-WiFi HetNet. 
We consider that voice users are always associated with LTE since unlike LTE, WiFi may not provide the required QoS for a voice user. 
However, data users can be associated with either LTE or WiFi. Offloading of data users from one RAT to another is considered at the time of association of 
voice users or departure of existing voice/data users. From a network operator's perspective, total system throughput is an important system 
metric since the generated revenue may largely depend on the number of bytes transported by the operator. 
Moreover, data users experiencing high throughput are more likely to adhere to a network operator, thus facilitating the improvement of the customer base of 
the operator. Therefore,
we aim to maximize the total system throughput and formulate this as a continuous-time Markov Decision Process (MDP) problem.\par
In the case of data users, although in low load condition, WiFi usually provides higher throughput than that of LTE, as the WiFi load increases, average per-user throughput in 
WiFi decreases rapidly \cite{g}. Therefore, under high WiFi load, LTE may offer more throughput than WiFi to data users and thus may be preferable to data users for the association.
However, voice and data users are allocated resources in LTE from a common resource block pool. The throughput requirement of LTE data users is usually more than 
that of the voice users. Therefore, maximization of the total system throughput may result in excessive blocking of voice users. The system may attempt to
save LTE resources which can be allocated later to data users having greater contributions to the system throughput than that of voice users. 
It results in an inherent trade-off between the total system throughput and the blocking probability of voice users. We consider this problem within the 
formalism of Constrained Markov Decision Process (CMDP), which maximizes the total system throughput subject to a constraint on the voice user blocking probability.\par
It is proved that the associated optimal policies contain a threshold structure, where after a certain threshold on the number of WiFi data users, data users are served using LTE. 
The existence of a similar threshold for the blocking of voice users is also established.
Based on the threshold based optimal policy,
we propose two RAT selection algorithms for LTE-WiFi HetNet. Extensive simulations are performed in ns-$3$ (a discrete event
network simulator) \cite{ns3} to evaluate the performance
of the proposed association algorithms. Using simulation results, performance gains of the proposed algorithms in comparison to another algorithm in the 
literature \cite{m} are also evaluated.
\subsection{Related Work}\label{sec:rw}
The solutions which investigate RAT selection problem in a HetNet, 
can be broadly divided into two categories, viz., user-initiated \cite{mikhail}-\cite{game} and network-initiated 
\cite{a}-\cite{q}.
In \cite{mikhail}, a user-initiated RAT selection algorithm based on Signal-to-Noise Ratio (SNR) 
and load information of individual
RATs with the adaptation of hysteresis mechanism, is considered for LTE-WiFi HetNet. The performance of this scheme is compared with
network-initiated cell-range extension schemes that use network-optimized Received Signal Strength Indicator (RSSI) bias value to steer users 
to other RATs. In \cite{farah}, a distributed RAT selection algorithm is proposed 
based on the distance and peak rate obtained from different IEEE 802.11 \cite{ieee} APs. \par
Few heuristic-based network-initiated RAT selection approaches are considered in \cite{song}-\cite{wlandata}. 
While the algorithm proposed in \cite{wlanfirst} prefers WLAN over cellular regardless
of the service type, the 
one proposed in \cite{wlandata} prefers cellular RAT for voice users and WLAN for data users.
Among the other network-initiated RAT selection schemes, \cite{a}-\cite{helou}, \cite{markovian}-\cite{q} consider various optimization 
approaches. In \cite{a}, optimal RAT selection problem is addressed in a HetNet to optimize throughput, blocking probability, etc.. Since the associated algorithm scales 
exponentially with the system size, authors also propose a computationally efficient heuristic policy.  
In \cite{h}, the association resulting in maximum value for the sum of logarithms of throughputs is chosen as
the optimal association among Wireless Stations (STAs) and APs. However, authors do not take into account user
arrival and departure. 
RAT selection policies in wireless networks \cite{altman1}-\cite{turhan1} are sometimes observed to contain certain threshold structures.
A multi-class admission system is considered in \cite{altman1}, where it is demonstrated that if it is optimal to accept a user of a class,
then it is optimal to accept a user of higher profit class too.\par
Offloading of data users from one RAT to another plays 
a major role in the capacity improvement of the system. Performance improvement achieved by on-the-spot offloading \cite{n}, a user-initiated WiFi 
offloading scheme, is analyzed in \cite{m}. The basic idea behind on-the-spot offloading is to steer the mobile data users to WiFi,
whenever WiFi is available. 
The user-initiated offloading 
scheme in \cite{i} is based on the combined information of signal strength and network load of LTE/WLAN. 
However, being a greedy one, this algorithm fails to converge to a globally optimum solution. 
The network-initiated offloading approach in \cite{b} computes the optimal fraction of traffic to be offloaded to WiFi
such that the per-user throughput of the system is maximized and performs better than 
on-the-spot offloading \cite{m}. 
However, the model in \cite{b}
does not incorporate voice users inside an LTE network. 
\subsection{Our Contribution}
In this paper, we investigate the optimal association policy in an LTE-WiFi HetNet. 
We consider a system where voice and data users can arrive or depart at any point in time. We introduce the possibility
of data user offloading from one RAT to another at the time of association or departure of a user.
We target to maximize the total system throughput.
The problem is formulated within the framework of
MDP. Another formulation is developed where we target to maximize the total system throughput, subject to a constraint on the voice user blocking probability, 
using CMDP.
Threshold structures of optimal policies are established.
We propose two algorithms based on the computed optimal policies and implement in ns-$3$.
3GPP recommended parameters are used in the simulations. Since most of the practical offloading schemes offload data users to WiFi,
performances of the proposed algorithms are compared with on-the-spot offloading algorithm \cite{m}.\par
The arrival of a new user in the 
LTE-WiFi system triggers the need for the optimal RAT selection.
Also, with the arrival or departure of users, 
the active users in different RATs may need to get offloaded to other RATs.
While few works in the literature have focused on RAT selection and offloading techniques, 
respectively, no existing literature, to the 
best of our knowledge, has addressed the issue of joint RAT selection and offloading for LTE-WiFi HetNet.\par
The rest of the paper is organized as follows. The system model is described in Section \ref{sec:sm}. In Section \ref{sec:pfsm}, 
the RAT selection problems 
are formulated within the framework of unconstrained and constrained continuous-time MDP, respectively. In Section \ref{sec:sop}, we derive the threshold structure of the optimal 
policy. Algorithms for the association 
of voice and data users in LTE-WiFi HetNet are proposed in Section \ref{sec:pniaa}. Section \ref{sec:nsr} presents simulation results. In Section 
\ref{sec:cfw}, we conclude the paper.

\section{System Model}\label{sec:sm}
We consider a system where an LTE BS and a WiFi AP are present. As illustrated in Fig.\ref{fig:hetnet},
we assume that 
both the BS and the AP are connected to a centralized 
controller by lossless links. 
We assume that the voice and data users are geographically located at any point
in the LTE BS coverage area. Since data users outside the dual coverage area of the LTE BS and the WiFi AP always get associated with the LTE BS and no decision
is involved in this case, without loss of generality, we take into consideration only those data users which are present inside the WiFi AP coverage area. 
We assume that there is a common resource pool in LTE for the voice users as well as the data users inside the WiFi AP coverage area. 
Data users inside the dual coverage area can be associated with the LTE BS or the WiFi AP. All the users are assumed to be stationary.
Voice and data user arrivals follow Poisson processes
with means $\lambda_v$ and $\lambda_d$, respectively. Service times for voice and data user are exponentially distributed with means $\frac{1}{\mu_v}$
and $\frac{1}{\mu_d}$, respectively. 
For justification behind these assumptions, see \cite{internet}.
\begin{remark}
Although for brevity of notation, a single LTE BS and a single AP have been considered, the system model can be generalized to a single LTE BS and 
multiple APs with non-overlapping coverage areas. 
Moreover, considering that each point in a geographical area is mapped to a single LTE BS (the LTE BS with highest average signal strength, say),
multiple BSs can also be included in the system model with slight modifications.
\end{remark}
\subsection{State Space}
We model the system as a controlled continuous time stochastic process $\{ X(t)\}_{t\ge 0}$ defined on a state space $\mathcal{S}$.
Any state $s\in \mathcal{S}$ is represented as a $3$-tuple 
${s}=(i,j,k),$ where $i,j$ and $k$ represent the number of voice users in LTE, the number of data users in LTE and the number of data users in WiFi, respectively.
The system state remains unchanged unless an existing user departs or a new user arrives in the system. The arrivals and departures in the system
are referred to as events. Five types of events are possible, viz., ($E_1$) an arrival of a new voice user in the system, 
($E_2$) an arrival of a new data user in the system, ($E_3$) a departure of an existing voice user from LTE, ($E_4$) a departure of an existing data user from LTE 
and ($E_5$) a departure of an 
existing data user from WiFi. Whenever an event occurs, the centralized controller takes an action, and based on the type of event and the action 
taken by the controller, a state transition may happen. Note that the transitions of $\{ X(t)\}_{t\ge 0}$ happen only at event epochs and not otherwise.
Thus, it suffices to observe the system state only at event epochs.
A finite amount of reward and cost are associated with every state-action pair.
Detailed descriptions of the action space, state transitions, reward and cost are provided in subsequent subsections.\par
Next, we elaborate on the structure of $\mathcal{S}$. We assume that $(i,j,k)\in \mathcal{S}$ if 
 $ (i+j) \le C$ and  $k \le W,$ 
where $C$ is
the total number of common resource blocks reserved in LTE for voice and data users, and $W$ is
the maximum number of users in WiFi, so that the per-user throughput in WiFi is greater than a threshold. 
The condition $(i+j) \le C$ arises because we assume that in each LTE subframe, every admitted user is allocated one resource block.
If this allocation is not possible, a new user is not admitted in the LTE system. Furthermore, note that WiFi throughput
decays monotonically \cite{g} as the number of WiFi users increases. We assume that each user gets more than a threshold value of average 
throughput (say $2$ Mbps), which leads to the bound $W$ on the maximum number of users that can be accommodated in the WiFi system. 
\begin{remark}Although the allocation of
multiple resource blocks is closer to the practical scenario, this complicates the system model
while the methodology and approach adopted in this paper do not change. 
\end{remark}
\subsection{Action Space}
The set of actions defines a set of possible association and offloading strategies in the event of arrival or departure of a user. 
Let the action space be denoted by $\mathcal{A}$.
Depending on the arrival or departure, we have a set of actions as stated below.
\[
    \mathcal{A}= 
\begin{cases}
    A_1,& \parbox[t]{.6\textwidth}{Block the arriving user or do nothing \\ during departure,} \\
    A_2,& \text{Accept voice/data user in LTE,} \\
    A_3,& \text{Accept data user in WiFi,} \\
    A_4,& \parbox[t]{.6\textwidth}{Accept voice user in LTE and offload \\one data user
    to WiFi,}\\
    A_5,& \parbox[t]{.6\textwidth}{Move one data user to a RAT
    (from \\ which departure has occurred).}\\
\end{cases}
\]
\begin{remark}
In this paper, actions are chosen based on the system state and the event occurred. One way of representing this is embedding the event in the state
space so that the action depends only on the system state. However, to avoid notational complications associated with this approach, we view the action 
as a function of the system state and the event.
\end{remark}
Let the set of states (subset of $\mathcal{S}$) in which action $a$ chosen based on an event $E_l$ is feasible
be denoted by ${\mathcal{S}}_{E_l,a}$.
Thus, in the case of voice user arrival, we have, 
\[
    {\mathcal{S}}_{E_1,a}= 
\begin{cases}
    \mathcal{S}\setminus \{(0,0,0)\},& a=A_1,\\
    \mathcal{S}\setminus \{(i,j,k):(i+j)=C\},& a=A_2, \\
    \mathcal{S}\setminus \{(i,j,k):(j=0) || (k=W)\},& a=A_4, \\
    \{\emptyset\},& \text{else}.
\end{cases}
\]

For data user arrival,
\[
    {\mathcal{S}}_{E_2,a}= 
\begin{cases}
     \{(i,j,W):(i+j)=C\},& a=A_1,\\
    \mathcal{S}\setminus \{(i,j,k):(i+j)=C\},& a=A_2, \\
    \mathcal{S}\setminus \{(i,j,k):k=W\},& a=A_3, \\
    \{\emptyset\},& \text{else}.
\end{cases}
\]

For voice user departure from LTE,
\[
    {\mathcal{S}}_{E_3,a}= 
\begin{cases}
    \mathcal{S}\setminus \{(i,j,k):i=0\},& a=A_1,\\
    \mathcal{S}\setminus \{(i,j,k):(i=0)||(k=0)\},& a=A_5, \\
    \{\emptyset\},& \text{else}.
\end{cases}
\]
For data user departure from LTE,
\[
    {\mathcal{S}}_{E_4,a}= 
\begin{cases}
    \mathcal{S}\setminus \{(i,j,k):j=0\},& a=A_1,\\
    \mathcal{S}\setminus \{(i,j,k):(j=0)||(k=0)\},& a=A_5, \\
    \{\emptyset\},& \text{else}.
\end{cases}
\]
For data user departure from WiFi,
\[
    {\mathcal{S}}_{E_5,a}= 
\begin{cases}
    \mathcal{S}\setminus \{(i,j,k):k=0\},& a=A_1,\\
    \mathcal{S}\setminus \{(i,j,k):(j=0)||(k=0)\},& a=A_5, \\
    \{\emptyset\},& \text{else}.
\end{cases}
\]
In the case of voice and
data user arrivals, the set of all possible actions are $\{A_1,A_2,A_4\}$ and $\{A_1,A_2,A_3\}$, respectively.
However, when an event $E_l$ occurs, action $a$ is not feasible if the system state is not present in ${\mathcal{S}}_{E_l,a}$.
In this paper, voice user blocking ($A_1$) 
is considered to be a feasible action in all the states, provided
the system is not empty. 
We consider blocking as a feasible action for data users, only when capacity is reached for both the RATs.
When a user departs from LTE or WiFi, the controller can choose either $A_1$ or $A_5$. If after the departure of a user from LTE, 
$A_5$ is 
chosen, it offloads one data user from WiFi to LTE. 
\subsection{State Transitions}
Based on an event and an action chosen,
from a state, the system moves deterministically to a different state. 
Assume that from the state $s=(i,j,k)$, the system moves to the state $s'(E_l,a)=(i',j',k')$
under the event $E_l$ and chosen action $a$. 
Values of $i',j'$ and $k'$ for different events $E_l$ (arrivals and departures of users) and action $a$ are tabulated in Table \ref{trpr}.
 \begin{table}[ht]
\caption{Transition Probability Table.}\label{trpr}
\centering 
\begin{tabular}{|l||l|}
\hline
{$\boldsymbol{(E_l,a)}$} & {$\boldsymbol{(i',j',k')}$} \\ \hline
$(\text{Arrival},A_1)$ & $(i,j,k)$ \\ \hline
$(\text{Voice departure from LTE},A_1)$ & $(i-1,j,k)$ \\ \hline
$(\text{Data departure from LTE},A_1)$ & $(i,j-1,k)$ \\ \hline
$(\text{Data departure from WiFi},A_1)$ & $(i,j,k-1)$ \\ \hline
$(\text{Voice arrival},A_2)$ & $(i+1,j,k)$ \\ \hline
$(\text{Data arrival},A_2)$ & $(i,j+1,k)$ \\ \hline
$(\text{Data arrival},A_3)$ & $(i,j,k+1)$ \\ \hline
$(\text{Voice arrival},A_4)$ & $(i+1,j-1,k+1)$ \\ \hline
$(\text{Voice departure from LTE},A_5)$ & $(i-1,j+1,k-1)$ \\ \hline
$(\text{Data departure from LTE},A_5)$ & $(i,j,k-1)$ \\ \hline
$(\text{Data departure from WiFi},A_5)$ & $(i,j-1,k)$ \\ \hline
\end{tabular}
\end{table}
Note that this table is exhaustive in all kinds of events and actions. However, in a state, we need to consider 
only those events and actions which are feasible in that state.
\subsection{Rewards and Costs}
Let the reward and cost functions per unit time corresponding to a state $s$, event $E_l$ and action $a$ be represented by $r(s,E_l,a)$ and $c(s,E_l,a)$, respectively. 
Let $R_{L,V}$ and $R_{L,D}$ denote the bit rate of voice and data users in LTE, respectively. To keep the model simple and computationally tractable, we assume 
that the bit rate of data users (e.g. data services like interactive video conferencing) in LTE is constant.
In general, a
voice user generates constant bit rate (CBR) traffic, and hence we take $R_{L,V}$ to be a constant. 
$R_{W,D}(k)$ corresponds to the per-user data throughput of $k$ users in WiFi. We assume full buffer traffic model \cite{g} for WiFi.
The calculation of $R_{W,D}(k)$ is based on the contention-driven medium access of WiFi users. It is a function of the probabilistic
transmission attempts of the users, corresponding
success and collision probabilities, and slot times for successful transmission, idle slots and busy
slots during collisions.
The reward per unit time in a state under the occurrence of an event and an action chosen is defined as the total system throughput in that state 
under that event and the chosen action.
For example, in the case of data user arrivals and $A_2$, it can be expressed as
\begin{equation*}
 r(s,E_2,A_2)= {iR_{L,V}+(j+1)R_{L,D}+kR_{W,D}(k)}.
\end{equation*}

The cost function considered here is as follows.
Whenever the centralized controller blocks an incoming voice user, one unit cost is incurred per unit time. Otherwise it is zero. 
Thus,
\begin{equation*}
c(s,E_l,a)=\begin{cases}
  1, & \text{if voice user is blocked}  , \\
  0, & \text{else}.\\
  \end{cases}
\end{equation*}
We consider blocking of data users only when both LTE and WiFi systems are full. 
Hence, we do not consider any cost on the blocking of data users.
\section{Problem Formulation And Solution Methodology}\label{sec:pfsm}
A \textit{decision rule} describes the mapping regarding which action is to be chosen at different states 
$s\in S$ and decision epochs $t_n$. An association \textit{policy} is a sequence of decision rules
$({\pi}^{t_1},{\pi}^{t_2},\ldots,{\pi}^{t_n},\ldots)$ taken at different decision epochs.  
Our goal is to determine an association policy which maximizes the total system throughput.
This can be formulated as a continuous-time unconstrained MDP problem.
In this case, a \textit{pure optimal policy} exists \cite{c}. 
Since the contribution of data users to the total system throughput is more than that of the voice users, the optimal association policy may result in 
high blocking probability of voice users. Hence, to address the trade-off between the total system throughput
and the voice user blocking probability, we consider the CMDP problem, 
where we target to maximize the total system throughput, subject to a constraint 
on the voice user blocking probability.
In this case, a stationary randomized optimal policy exists \cite{p}. A \textit{Randomized
policy} is a mixture of two pure policies with associated probabilities. 
Arrival and departure of users can occur at arbitrary points in time, which makes the problem
continuous time in nature.
\subsection{Problem Formulation} 
Let $\mathcal{M}$ be the set of all memoryless policies. To guarantee a unique stationary distribution, we assume that Markov chains associated with 
such policies are irreducible. 
Following the policy $M\in \mathcal{M}$, let the average reward and cost of the system over infinite horizon be denoted by
$V^{M}$ and $B^{M}$, respectively. Let $R(t)$ and $C(t)$ be the total reward and cost of the system incurred up to time $t$, respectively.
For the unconstrained MDP problem, our objective is to maximize the total system throughput which can be described as follows.
\begin{equation}\label{max1}
\text{Maximize:}\quad V^{M}=\lim_ {t\to \infty} {\frac{1}{t}{\mathbb{E}_{M} [R(t)]}},
\end{equation}\\
where $\mathbb{E}_{M}$ denotes the expectation operator under the policy $M$.
However, for the CMDP problem, our objective is to maximize the total system throughput,
subject to a constraint on the blocking probability of voice users. This can be described as follows.
\begin{equation}\label{max2}
\begin{split}
 &\text{Maximize:} \quad V^{M}=\lim_ {t\to \infty} {\frac{1}{t}{\mathbb{E}_{M} [R(t)]}},\\&
 \text{Subject to:} \quad B^{M}=\lim_ {t\to \infty} {\frac{1}{t}{\mathbb{E}_{M} [C(t)]}}\le B_{\max},
\end{split}
\end{equation}
where $B_{\max}$ denotes the constraint on the blocking probability of voice users.  
Our objective is to determine the optimal policy for both unconstrained and constrained MDP problem. 
Since the optimal policies are known to be stationary policies, the corresponding limits in Equation (\ref{max1}) and (\ref{max2}) exist.
\subsection{Conversion to Discrete-Time MDP}
To compute the optimal policy, we can use the well-known Value Iteration Algorithm (VIA) \cite{c}. 
However, before that, the continuous-time MDP has to be transformed into an equivalent discrete-time MDP using uniformization \cite{c}, 
so that both models have the same average expected reward and cost for a stationary policy.\par Let $\tau_s(E_l,a)$ represent the expected time until
the next event, if action $a$ is chosen in state $s$ under the event $E_l$.  We need to choose a number $\delta$, such that
$0<\delta \ll \displaystyle \min_{s,E_l,a} \tau_s(E_l,a)$. The state space and the action space remain the same
in the equivalent discrete-time model. 
Let $\hat{p}(s,E_l)$, $\hat{r}(s,E_l,a)$ and $\hat{c}(s,E_l,a)$ represent the probabilities of the event, reward and cost in the transformed discrete-time model 
in state $s$ under the event $E_l$ and action $a$, respectively. 
Thus, we have,
\begin{equation*}
\hat{r}(s,E_l,a)=r(s,E_l,a)\quad \text{and}\quad \hat{c}(s,E_l,a)=c(s,E_l,a). 
\end{equation*} 
$\hat{p}(s,E_l)$ is a function of rate of different events $E_l$ and $\delta$. 
%
%
Note that, this discrete-time MDP has identical stationary policies to that of the continuous-time MDP.
\subsection{Lagrangian Approach}
After conversion into an equivalent discrete-time MDP model, 
we use the Lagrangian approach \cite{p} to solve the CMDP.
For a fixed value of Lagrange Multiplier (LM) $\beta$, 
the modified reward function of the CMDP is
\begin{equation*}
\hat{r}(s,E_l,a;\beta)=\hat {r}(s,E_l,a)-\beta \hat {c}(s,E_l,a). 
\end{equation*}
The dynamic programming equation below describes the necessary condition for optimality.
\begin{equation}\label{opt1} 
\begin{split}
V(s)=&\sum \limits_{E_l}\hat{p}(s,E_l) \max \limits_a [\hat{r}(s,E_l,a;\beta)+V(s')]+\big(1-\sum \limits_{E_l} \hat{p}(s,E_l)\big)V(s),
\end{split}
\end{equation}
where $V(s)$ denotes the value function in state $s\in \mathcal{S}$. 
Next, our aim is to determine the value of $\beta$ ($=\beta^*$, say) which maximizes the average expected reward, subject to a cost constraint. 
The value of $\beta^*$ can be determined using gradient descent algorithm, as discussed in \cite{e}. In $k$th iteration, we modify the value of $\beta$ from its 
previous iteration as,
\begin{equation}\label{betago}
 \beta_{k+1}=\beta_{k}+\frac{1}{k}(B^{\pi_{\beta_k}} -B_{\max}).
\end{equation}
For a fixed value of $\beta$, the unconstrained maximization problem can be solved using VIA, as described below. 
\begin{equation}\label{opt}
\begin{split}
V_{n+1}(s)=&\sum \limits_{E_l}\hat{p}(s,E_l) \max \limits_a [\hat{r}(s,E_l,a;\beta)+V_n(s')]+\big(1-\sum \limits_{E_l} \hat{p}(s,E_l)\big)V_n(s),
\end{split}
\end{equation}
where $V_n(.)$ is an estimate of the value function after $n$th iteration. 
After determining $\beta^*$, the next step is to determine the optimal policy for the CMDP problem. As discussed
in \cite{p}, optimal policy for a CMDP problem is a mixture of two pure policies $\pi_{\beta^*-\epsilon}$ and $\pi_{\beta^*+\epsilon}$,
obtained by perturbation of $\beta^*$ by a small amount $\epsilon$ in both directions. Let the long-term average 
expected costs of the two pure policies be $B_{\beta^*-\epsilon}$ 
and $B_{\beta^*+\epsilon}$, respectively. In the next step, we determine the value of the parameter $p$ 
such that 
\begin{equation*}
pB_{\beta^*-\epsilon}+(1-p)B_{\beta^*+\epsilon}=B_{\max}. 
\end{equation*} 
Finally, we have a randomized optimal policy for the considered CMDP problem. 
At each decision epoch, policies $\pi_{\beta^*-\epsilon}$ and $\pi_{\beta^*+\epsilon}$ 
are chosen w.p. $p$ and $(1-p)$, respectively. \par
Note that the iterations on LM described above are necessary only for the CMDP problem. In the case of unconstrained MDP, 
VIA can be employed to determine a pure optimal policy,
after an equivalent
discrete-time MDP model is obtained.
\section{Structure of the Optimal Policy}\label{sec:sop}
The dynamic programming equations (Equation (\ref{opt1}) and (\ref{opt})) described in the previous section are exploited to establish the fact that the optimal
policy is of threshold type. The optimality of threshold policy is established with the aid of some lemmas presented below. For the purpose 
of readability, we present the proofs of the lemmas in Appendices.
\subsection{Optimal Policy for Data Users}
In this section, we present structural properties on the optimal policy for the service of data users along with their physical interpretations.
Let us denote the throughput increment in WiFi when the number of WiFi users increases from $k$ to $(k+1)$ by $\tilde {R}_{W,D}(k)$.
Therefore,
$\tilde {R}_{W,D}(k)=(k+1){R}_{W,D}(k+1)-k {R}_{W,D}(k)$. We assume the following.
\begin{assumption}\label{ass:1}
Let $R_{L,D}$ be such that $R_{L,D}\ge \tilde {R}_{W,D}(k)$, $\forall k \ge k_{th}$
and $R_{L,D}< \tilde {R}_{W,D}(k), \forall k < k_{th},$
where $k_{th}$ is a threshold such that if $k\ge k_{th}$, the data rate improvement provided
by a single data user in the LTE system is more than the improvement in total WiFi throughput as the number of WiFi data users 
is increased from $k$ to $(k+1)$.
\end{assumption}
\begin{remark} 
Following the full buffer traffic model \cite{g}, $\tilde {R}_{W,D}(k)$ initially increases with $k$
and then decreases. This behavior matches with Assumption \ref{ass:1}. 
\end{remark}
The following two lemmas describe a threshold structure on the optimal policy for the service of data users.
Specifically, up to a certain threshold on the total number of data users, data users are served using WiFi. After the threshold is crossed, 
data users are served using LTE.
\begin{lemma}\label{lemma1}
For every $i$ and $j$ such that $(i+j)<C$, if the total number of data users in the system is $(j+k)\le k_{th}$, then the optimal policy is to serve all data users
using WiFi. In other words,$(j+k)\le k_{th} \implies j=0$.
\end{lemma}
\begin{proof} 
See Apppedix \ref{app:a}.
\end{proof}
In Lemma \ref{lemma1}, following Assumption \ref{ass:1}, since for $k< k_{th}$, the data rate improvement is more if an additional data user is served using
WiFi rather than using LTE, it is optimal to serve the data users using WiFi.
\begin{lemma}\label{lemma2}
For every $i$ and $j$ such that $(i+j)<C$, if the total number of data users in the system is $(j+k)> k_{th}$, then the optimal policy is to serve $k_{th}$ data
users using WiFi and all other data users
using LTE. In other words,$(j+k)> k_{th} \implies k=k_{th}$.
\end{lemma}
\begin{proof} 
See Apppedix \ref{app:b}.
\end{proof}
The physical significance of Lemma \ref{lemma2} is that for $k\ge k_{th}$, the data rate improvement provided
by a single data user in LTE is more than that of the WiFi (following Assumption \ref{ass:1}), and hence
it is optimal to serve up to $k_{th}$ data users using WiFi and serve the additional data users using LTE.\par
Following lemma is a direct consequence of how the system is modeled. 
\begin{lemma}\label{lemma3}
For every $i$ and $j$ such that $(i+j)=C$, the optimal policy is to serve all the incoming data users using WiFi until $k=W$, where an incoming data user is 
blocked.
\end{lemma}
\subsection{Optimal Policy for Voice Users}
In this section, we characterize the optimal policy for the arrival of voice users. We prove that the optimal policy is of threshold type.
Let $D_i$ be the difference operator which is defined as $D_iV(i,j,k)=V(i+1,j,k)-V(i,j,k)$. Similarly, we define the second difference operator
as $D_{ii}(.)=D_i(D_i(.))$. 
Let $E_i$ be another difference operator defined as $E_iV(i,j,k)=V(i+1,j-1,k+1)-V(i,j,k)$. We define the second difference operator
as $E_{ii}(.)=E_i(E_i(.))$. Similarly, we define $F_iV(i,j,k)=V(i+1,j-1,k)-V(i,j,k)$.
In this section, the terminologies ``increasing'' and ``decreasing'' are used in the weak sense of ``non-decreasing'' and ``non-increasing'',
respectively.
In each state, let the sum of arrival and service rates be denoted
by $v(i,j,k)$. 
Thus, we have,
\begin{equation*}
v(i,j,k)=\lambda_v+\lambda_d+i\mu_v+j\mu_d+k\mu_d.                                                                       
\end{equation*}
Let us define $f(i,j,k)=(iR_{L,V}+jR_{L,D}+kR_{W,D}(k))$. 
The lemma presented below describes the superiority of one action over the other for the service of incoming voice users. Specifically, up
to a certain threshold on the total number of data users, $A_4$ (accept voice user in LTE with data user offload to WiFi) is better than $A_2$ 
(accept voice user in LTE). After the threshold is crossed,
$A_2$ becomes better.
\begin{lemma}\label{lemma4}
In the case of a voice user arrival in a state $(i,j,k)$, where
$(i+j)<C$,
\begin{itemize}
\item [(i)] $A_4$ is always better than $A_2$ if $k< k_{th}$, 
\item [(ii)] $A_2$ is always better than $A_4$ if $k\ge k_{th}$.
\end{itemize}
\end{lemma}
\begin{proof} 
Proof is similar to the proof of Lemma \ref{lemma1}.
\end{proof}
Similar to Lemma \ref{lemma1} and \ref{lemma2}, following Assumption \ref{ass:1}, since for $k< k_{th}$, the data rate improvement is more if an additional data user 
is served using WiFi rather than using LTE, $A_4$ is better than $A_2$.
Therefore, when $k< k_{th}$, the choice of optimal action is between $A_4$ (accept voice user in LTE with data user offload to WiFi) and $A_1$ (blocking). 
Similarly, for $k\ge k_{th}$, the optimal action is either $A_2$ (accept voice user in LTE)
or $A_1$.\par
The following lemma describes that when capacity is not reached in LTE and a voice user arrives, a threshold structure is observed. Until a threshold on the number of 
voice/data users in LTE,
$A_2$ (for $k\ge k_{th}$) or $A_4$ (for $k< k_{th}$) is preferred. After the threshold $A_1$ becomes optimal.
\begin{lemma}\label{lemma5}
For every $i$ and $j$ such that $(i+j)<C$ and a voice user arrival,
\begin{itemize}
\item [(i)] if the optimal action in state $(i,j,k)$ is $A_1$, then the optimal action in state $(i+1,j,k)$ and in state $(i,j+1,k)$ 
is also $A_1$,
\item [(ii)] 
if the optimal action in state $(i,j,k)$ is $A_2$ ($A_4$), then the optimal action in state $(i-1,j,k)$ and in state $(i,j-1,k)$ 
is also $A_2$ ($A_4$).
\end{itemize}
\end{lemma}
\begin{proof} See Apppedix \ref{app:c}.
\end{proof}
When the number of voice/data users in LTE is less, $A_2$ or $A_4$ is chosen as the optimal action in the event of a voice user arrival.
When $i$ or $j$ crosses a certain threshold, the number of free resources for incoming voice users decreases. Therefore, the blocking probability of voice users
increases. Thus, after a threshold on $i$ or $j$, $A_1$ becomes optimal.\par
However, when $(i+j)=C$, since $A_2$ is infeasible, optimal action is either $A_1$ or $A_4$.
The lemma presented below discusses the threshold nature of the optimal policy for voice user arrivals when $(i+j)=C$.
\begin{lemma}\label{lemma6}
For every $i$ and $j$ such that $(i+j)=C$ and a voice user arrival,
\begin{itemize}
\item [(i)]  if the optimal action in state $(i,j,k)$ is $A_1$, then the optimal action in 
state $(i+1,j-1,k)$ is also $A_1$,
\item [(ii)] if the optimal action in state $(i,j,k)$ is $A_4$, then the optimal action in
state $(i-1,j+1,k)$ is also $A_4$.
\end{itemize}
\end{lemma}
\begin{proof} See Apppedix \ref{app:d}.
\end{proof}
The physical interpretation of the above lemma is that for states with $(i+j)=C$, when $i$ is small, $A_4$ is preferred. However, when $i$ crosses a 
threshold, since $j$ becomes small, and consequently the total system 
throughput is small, $A_4$ may further lower the total system throughput. Therefore, blocking of voice users is chosen as the 
optimal action.
\section{Proposed Network-Initiated Association Algorithm}\label{sec:pniaa}
Based on the optimal policy computed by solving the unconstrained MDP and CMDP problem, respectively, in this section, 
we propose two network-initiated association algorithms for LTE-WiFi HetNet. 
The details of the unconstrained MDP-based algorithm is presented below. \par
\begin{algorithm}[t]
\caption{Network-Initiated Unconstrained MDP-based Association Algorithm.}
\label{algo3}
\begin{algorithmic}[1]
\Require $\lambda_v,\lambda_d,\mu_v,\mu_d,R_{L,V},R_{L,D},R_{W,D}(.).$
\State Compute threshold $k_{th}$ for the association of data users.
\Procedure{Calc\textendash Opt\textendash Policy\textendash uc}{}
	\State Calculate optimal policy using VIA (Equation (\ref{opt})).
\EndProcedure
\Ensure Deterministic optimal policy.
\State Compute thresholds $va_c(j,k)$ and $va_{lc}(j,k)$ for the association of voice users for $(i+j)=C$ and $(i+j)<C$, respectively.
\Procedure{Policy\textendash Impl}{}
\For{each arrival of voice users}
	
	 
	 \If {$(i+j)<C$} 
	 \State Choose $A_1$ if $i \ge va_{lc}(j,k)$.
	 \State Choose $A_2$ if {$i< va_{lc}(j,k)$ and $k\ge k_{th}$}.
	 \State Choose $A_4$  otherwise.
\Else
	 \State Choose $A_4$ if $i< va_c(j,k)$, $A_1$ otherwise.
	 \EndIf

\EndFor
\For{each arrival of data users}
	
	\If {$(i+j)<C$}
	 \State Choose $A_3$ if $k< k_{th}$, $A_2$ otherwise. 
         \Else
	 \State Choose $A_3$ if $k<W$, $A_1$ otherwise.
	 \EndIf

\EndFor
\For{each departure of users from LTE (WiFi)}
	
	 \State Choose $A_1$ ($A_5$) if {$k\le k_{th}$}, $A_5$($A_1$) otherwise.
\EndFor
%

\EndProcedure
\end{algorithmic}
\end{algorithm}
The procedure CALC\textendash OPT\textendash POLICY\textendash UC in Algorithm \ref{algo3} computes the optimal policy by solving an unconstrained MDP problem using
solution methodologies described in Section \ref{sec:pfsm}. 
The calculated thresholds for the association of voice/data users are made available to the centralized controller
connected to both the LTE BS and the WiFi AP. 
Since the centralized controller has an overall view of the whole system, 
the information regarding the numbers of active voice and data users in LTE and WiFi, are available to it.
Whenever there is an arrival or a departure, the controller initiates the procedure 
POLICY\textendash IMPL, as described in Algorithm \ref{algo3}. This procedure determines the state of the system based on the number of 
active users in LTE and WiFi networks and then chooses the appropriate action based on the corresponding thresholds. \par
Next, we describe the CMDP-based algorithm which addresses the issue of high blocking probability of voice users, which may be encountered in Algorithm \ref{algo3}.
The algorithm is described in detail below.\par
\begin{algorithm}[t]
\caption{Network-Initiated Constrained Association Algorithm.}
\label{algo4}
\begin{algorithmic}[1]

\Require $\lambda_v,\lambda_d,\mu_v,\mu_d,R_{L,V},R_{L,D},R_{W,D}(.),B_{\max}.$
\State Compute threshold $k_{th}$ for the association of data users.
\Procedure{Calc\textendash Opt\textendash Policy}{}
\State Initialize $\beta$.
\While {{$B^{\pi_{\beta}} \neq B_{\max}$}}
	\State Calculate the optimal policy using VIA.
\State Update $\beta$ using Equation (\ref{betago}).
\EndWhile
\EndProcedure
\Ensure Randomized optimal policy.
\State Compute thresholds $va_c(j,k)$ and $va_{lc}(j,k)$ for association of voice users for $(i+j)=C$ and $(i+j)<C$, respectively.
\Procedure{Policy\textendash Impl}{}
\State As discussed in Algorithm \ref{algo3}.
\EndProcedure
\end{algorithmic}
\end{algorithm}
Apart from the same set of input parameters as required by Algorithm \ref{algo3}, Algorithm \ref{algo4} requires $B_{\max}$ as an additional parameter 
to specify the constraint on the blocking probability of voice users.
The procedure CALC\textendash OPT\textendash POLICY in Algorithm \ref{algo4} computes the randomized optimal policy for the considered CMDP problem. First, the optimal
policy is determined for a fixed value of $\beta$, and then the value of $\beta$ is updated until $B^{\pi_{\beta}}$ becomes equal to $B_{\max}$. All
other procedures are similar to the procedures described in Algorithm \ref{algo3}. 
\section{Numerical and Simulation Results}\label{sec:nsr}
 In this section, the algorithms proposed in the last section are
implemented in ns-$3$ to observe the performance of the proposed algorithms.
Performance of the proposed algorithms in terms of blocking probability of voice users and the total system throughput is compared to the performance of on-the-spot WiFi offloading algorithm \cite{m}.
In this algorithm \cite{m}, data user chooses LTE, only when there is no WiFi coverage. Therefore, in the considered system model, 
with on-the-spot offloading, data users always get associated with WiFi 
until WiFi capacity is exhausted. Voice users always get associated with LTE BS, and when LTE capacity is full, 
they are blocked. 
\subsection{Simulation Model and Evaluation Procedure}
The simulated network model consists of a $3$GPP LTE BS and an IEEE $802.11$g WiFi AP. 
All users are taken to be stationary. The AP is approximately $50$ m away from
the LTE BS, and data users are distributed uniformly within $30$ m radius of the WiFi AP.
The WiFi AP is assumed to be deployed by the same cellular operator and hence trusted from the point of view of interworking. 
LTE and WiFi network parameters used in the simulation, as summarized in Table \ref{tablee1} and \ref{tablee2}, are based on 3GPP models \cite{r}-\cite{s}  
and saturation throughput \cite{g} $802.11$g WiFi model.
Propagation delay in WiFi network is assumed to be negligible.
We consider CBR traffic for voice and data users in LTE. The generation of a fixed rate uplink flow 
is implemented in ns-$3$ using an application developed by us, which works similar to the ON/OFF application.
This application creates sockets between the sender and the receiver, and fixed sized packets are transmitted from the sender to the receiver at a
constant bit rate. 
\begin{table}[ht]
\caption{LTE Network Model.}\label{tablee1}
\centering 
\begin{tabular}{|l||l|}
\hline
\textbf{Parameter} & \textbf{Value} \\ \hline
Maximum voice capacity & $10$ users \\ \hline
Maximum data capacity & $10$ users \\ \hline
Voice bit rate of a single user & $20$ kbps \\ \hline
Data bit rate of a single user & $5$ Mbps \\ \hline
Voice packet payload & $50$ bits \\ \hline
Data packet payload & $600$ bits \\ \hline
Tx power for BS and MS & $46$ dBm and $23$ dBm \\ \hline
Noise figure for BS and MS & $5$ dB and $9$ dB\\ \hline
Antenna height for BS and MS & $32$ m and $1.5$ m\\ \hline
Antenna parameter for BS and MS & Isotropic antenna \\ \hline
Path loss & $128.1+37.6\log(R)$, $R$ in kms\\ \hline
\end{tabular}
\end{table}
\begin{table}[ht]
\caption{WiFi Network Model.}\label{tablee2}
\centering 
\begin{tabular}{|l||l|}
\hline
\textbf{Parameter} & \textbf{Value} \\ \hline
Channel bit rate & $54$ Mbps \\ \hline
UDP header & $224$ bits \\ \hline
Packet payload & $1500$ bytes \\ \hline
Slot duration & $20$ $\mu$s \\ \hline
Short inter-frame space (SIFS) & $10$ $\mu$s \\ \hline
Distributed Coordination Function IFS (DIFS) & $50$ $\mu$s \\ \hline
Minimum acceptable per-user throughput & $3.5$ Mbps \\ \hline
Tx power for AP & $23$ dBm \\ \hline
Noise figure for AP & $4$ dB\\ \hline
Antenna height for AP & $2.5$ m\\ \hline
Antenna parameter & Isotropic antenna \\ \hline\end{tabular}
\end{table}
\subsection{Voice User Arrival Rate Variation}
\subsubsection{Voice User Blocking Probability Performance}
Fig. \ref{fig:bfvoice} illustrates the variation of voice user blocking percentage of on-the-spot offloading \cite{m}, Algorithm \ref{algo3} and \ref{algo4} 
as a function of $\lambda_v$. 
\begin{figure*}[t!]
    \centering
    \begin{subfigure}[t]{0.30\textwidth}
        \centering
        \includegraphics[width=\textwidth]{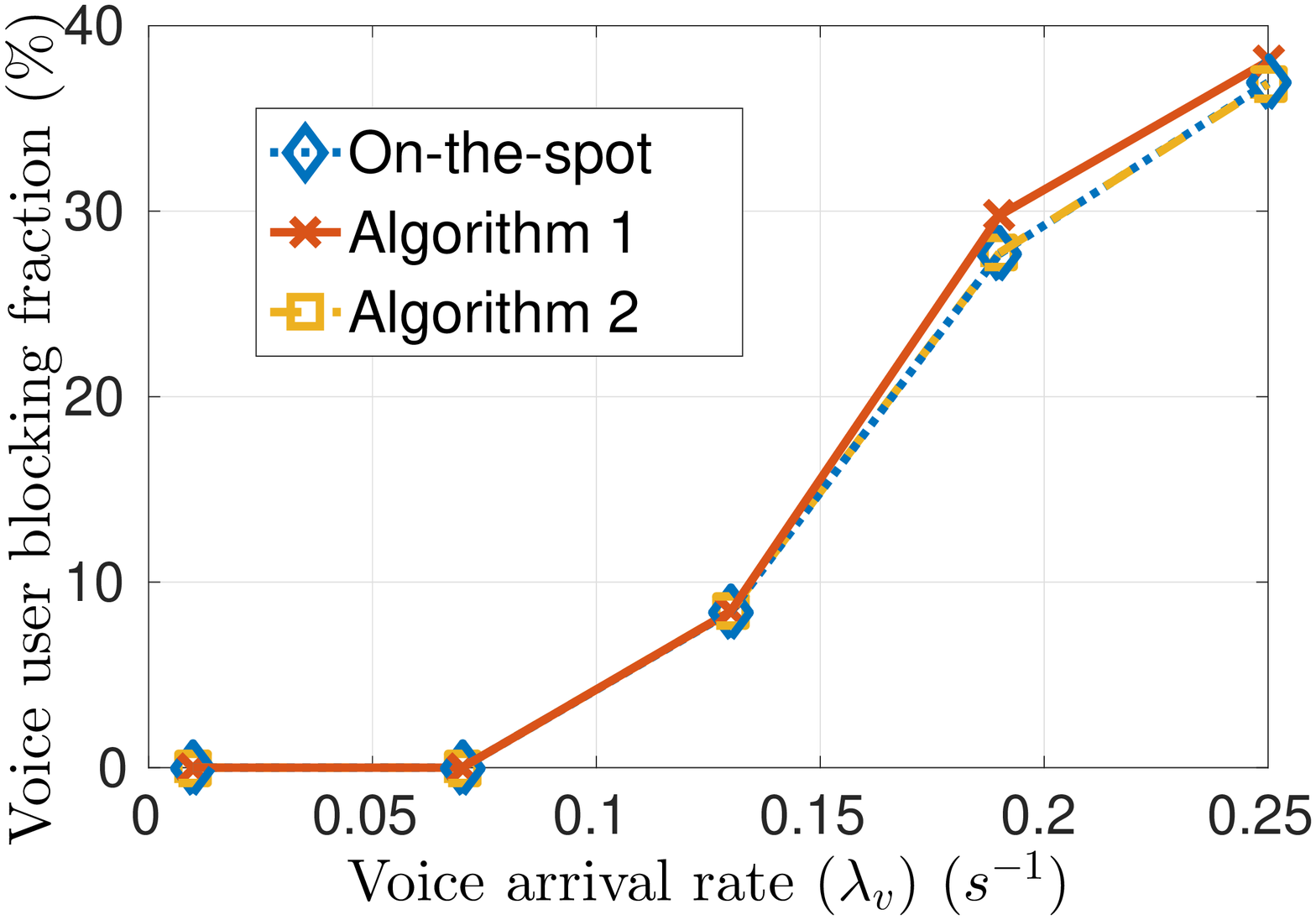}
        \caption{Voice user blocking percentage vs. $\lambda_v$ ($\lambda_d=1/20,\mu_v=1/60$ and $\mu_d=1/10$).}
        \label{fig:bfvoice}
        \end{subfigure}%
    ~ 
    \begin{subfigure}[t]{0.30\textwidth}
        \centering
        \includegraphics[width=\textwidth]{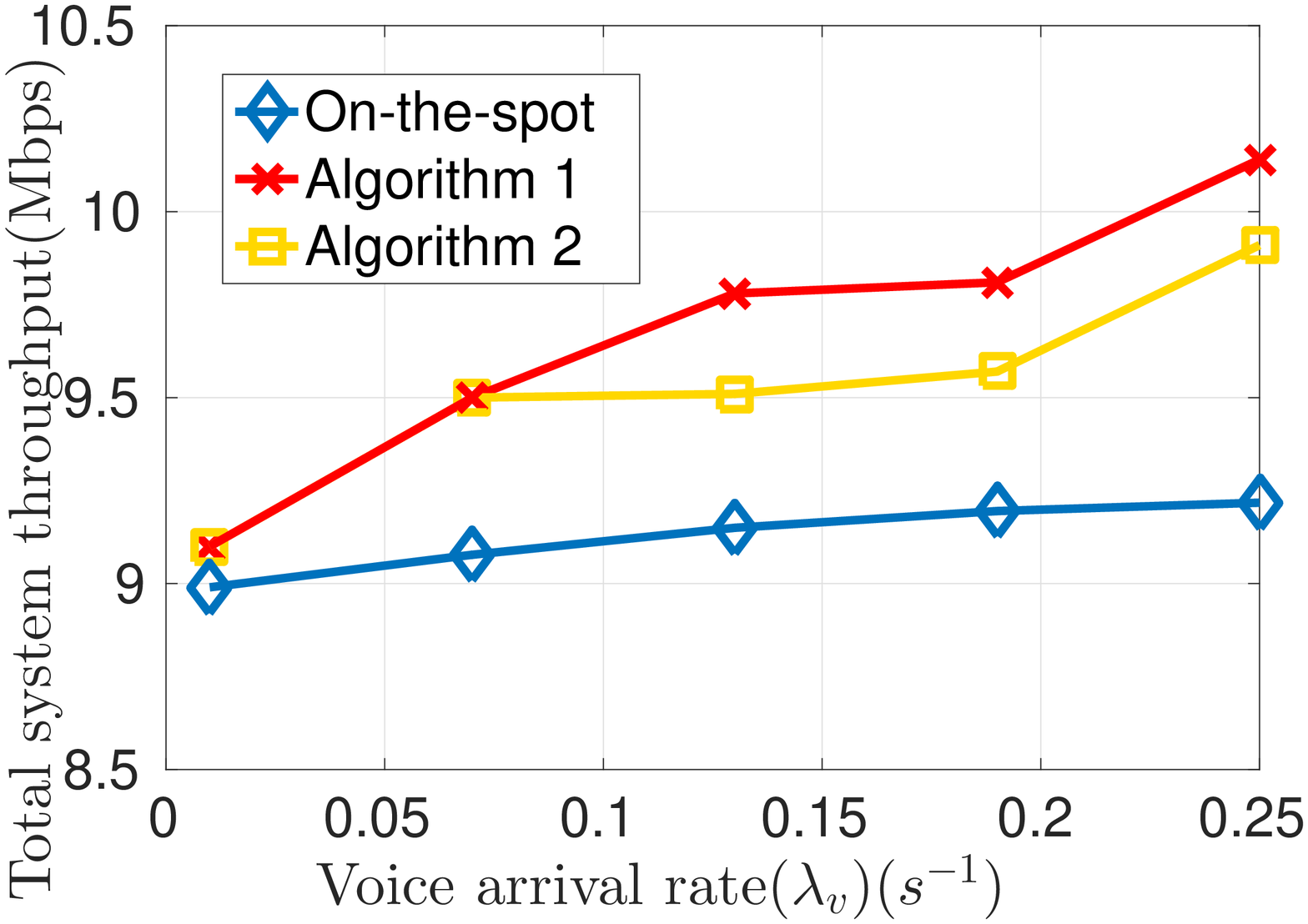}
        \caption{Total system throughput vs. $\lambda_v$ ($\lambda_d=1/20,\mu_v=1/60$ and $\mu_d=1/10$).}
        \label{fig:putvoice}
        \end{subfigure}
        ~ 
    \begin{subfigure}[t]{0.30\textwidth}
        \centering
        \includegraphics[width=\textwidth]{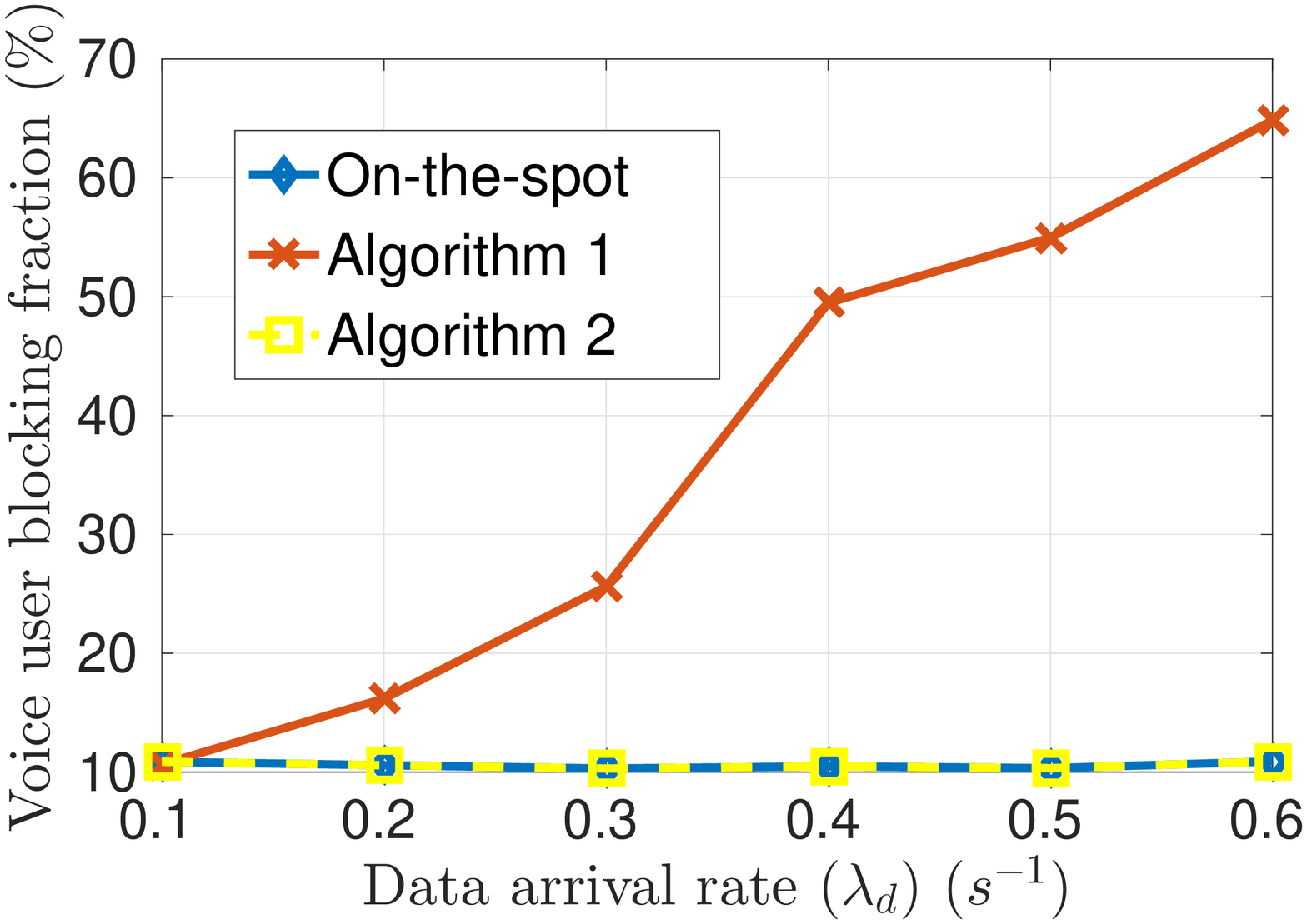}
         \caption{Voice user blocking percentage vs. $\lambda_d$ ($\lambda_v=1/6,\mu_v=1/60$ and $\mu_d=1/10$).}
         \label{fig:bfdata1}
       \end{subfigure}

\caption{Plot of blocking fraction of voice users and total system throughput for different algorithms.}
\end{figure*}
%
In on-the-spot offloading, voice users are blocked when LTE reaches the capacity. 
When $\lambda_v$ is small, 
the voice user blocking probability is small. However, as $\lambda_v$ increases, the
probability of approaching the LTE capacity and hence the voice user blocking probability increases.
The voice user blocking probability in Algorithm \ref{algo3} is small when $\lambda_v$ is small. However, as $\lambda_v$ increases, voice user 
blocking probability values become marginally higher than the corresponding values for on-the-spot offloading.
Algorithm \ref{algo3} may introduce 
blocking of voice users even when LTE has not reached its capacity, i.e., for states with $(i+j)<C$.
Voice users have very less contribution to the total system 
throughput. Hence, voice users are blocked to save resources for data users which contribute significantly
to the total system throughput. 
However, in Algorithm \ref{algo4}, the number of states with proactive blocking (blocking when $(i+j)<C$) is reduced due to a constraint on the voice user 
blocking probability.
Additionally, when $i$ is small, the optimal action in states with $(i+j)=C$ becomes $A_4$ (accept voice user in LTE and 
data offload to WiFi). Voice user blocking probability contribution comes mainly from the states with $(i+j)=C$,
where $i$ is large (say states $(C,0,0)$,$(C-1,1,0)$ etc.). 
Since a major fraction of voice user blocking occurs when $(i+j)=C$ and $i$ is large, 
the system becomes analogous to the on-the-spot
offloading. Hence, the voice user blocking probability performance of Algorithm \ref{algo4} is almost similar to on-the-spot offloading algorithm \cite{m}. 

\subsubsection[\textwidth=14 cm]{Total System Throughput Performance}
Total system throughput performance comparison of different algorithms is illustrated in Fig. \ref{fig:putvoice}.
In on-the-spot offloading, the average number of voice users
in LTE increases with $\lambda_v$, while the average number of data users in WiFi remains constant. Thus, the total system throughput increases with $\lambda_v$.
For Algorithm \ref{algo3}, 
with an increase in $\lambda_v$, the blocking probability of voice users increases. 
Therefore, the fraction of voice users
in the system decreases, and the total system throughput increases. Besides, Algorithm \ref{algo3} performs a 
significant amount of load balancing under
$A_4$ (accept voice user in LTE with data user offload to WiFi) and $A_5$ (move data user to the RAT from where a user has departed). 
With higher $\lambda_v$, load balancing actions are chosen more frequently.
Thus, with higher $\lambda_v$, Algorithm \ref{algo3}
exhibits greater improvement over on-the-spot offloading algorithm.
The improvement in total system throughput varies from $1.22\%$ 
(for $\lambda_v=0.01$) to $10.32\%$ (for $\lambda_v=0.25$).
In Fig. \ref{fig:putvoice}, we observe that Algorithm \ref{algo4} also performs better than on-the-spot
offloading. However, due to a constraint on the voice
user blocking probability,
performance improvement is lower than Algorithm \ref{algo3}. 
For lower values of $\lambda_v$ ($\lambda_v=0.01,0.07$), the total system throughput of Algorithm \ref{algo4} is same as that of Algorithm \ref{algo3}
as the optimal policy
for the CMDP is same as that of the unconstrained MDP. 
On-the-spot offloading algorithm blocks voice users only when LTE reaches capacity. Typically, in Algorithm \ref{algo4} also, 
voice user blocking occurs when the LTE is full with a large number of voice users. 
However, due to load balancing of data users,
Algorithm \ref{algo4} outperforms the on-the-spot offloading algorithm. 
With $\lambda_v=0.01$, the improvement in total system throughput is only $1.22\%$ and with $\lambda_v=0.25$, it becomes $7.60\%$.
\subsection{Data User Arrival Rate Variation}
\subsubsection{Voice User Blocking Probability Performance}
In Fig. \ref{fig:bfdata1}, 
for on-the-spot offloading, 
voice and data 
users are accepted in LTE and WiFi, respectively. Consequently, changes in $\lambda_d$ do not affect the blocking probability performance of 
voice users in LTE.
In the case of Algorithm \ref{algo3}, increase in $\lambda_d$ associates more number of data users with LTE, since the optimal policy for data users is to associate
with LTE after the number of WiFi data users crosses a certain threshold. Therefore, the number of free LTE resources for voice users reduces, eventually
increasing the blocking probability of voice users. The voice user blocking probability of Algorithm \ref{algo3} is worse than that of on-the-spot offloading and 
increases with $\lambda_d$. 
The blocking probability performance of Algorithm \ref{algo4} is similar to that of the on-the-spot offloading.
Since usually the voice users are blocked in the states where the only feasible action is blocking (say state $(C,0,0)$), the decision epochs where voice users 
are blocked are almost same as that of the on-the-spot offloading. 
\subsubsection{Total System Throughput Performance}
In Fig. \ref{fig:putdata}, total system throughputs for different algorithms are plotted as a function of $\lambda_d$. 
\begin{figure}[h]
 \begin{center}
\includegraphics[width=0.42\textwidth]{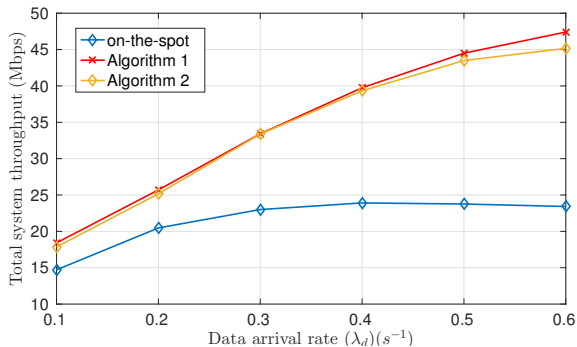}
\caption{Total system throughput vs. $\lambda_d$ ($\lambda_v=1/6,\mu_v=1/60$ and $\mu_d=1/10$).}
\label{fig:putdata}
 \end{center}
\end{figure}
In on-the-spot offloading, 
with an increase in $\lambda_d$, the number of WiFi data users increases, and this increases the total system throughput.
However, for high $\lambda_d$, the effect of contention among data users reduces the rate of increment of the total system throughput.
In Algorithm \ref{algo3}, as $\lambda_d$ increases, more number of data users are served using LTE. Since the throughput contribution of 
data users is more than voice users, the blocking probability of voice users increases with $\lambda_d$. Thus, the fraction of voice users in the system reduces,
effectively causing more improvement in the total system throughput.
When $\lambda_d=0.1$, the improvement in system metric is $25.22\%$, whereas for $\lambda_d=0.6$,
the system metric almost doubles.
In Fig. \ref{fig:putdata}, 
the total system throughput values for Algorithm \ref{algo4} are smaller 
than the corresponding 
values for Algorithm \ref{algo3}. The reduction in blocking probability of voice users comes 
at a price of the reduction in the total system throughput. 
Still, due to optimal association and load balancing decisions,
Algorithm \ref{algo4} reduces the effect of contention among data users in WiFi 
and hence performs better than on-the-spot offloading algorithm \cite{m}. 
For example, with $\lambda_d=0.1$, the improvement in system metric is about $22.96\%$ and with $\lambda_v=0.6$, it becomes almost $93\%$.

\section{Conclusion}\label{sec:cfw}
In this paper, we have formulated the optimal association problem in an LTE-WiFi HetNet as an MDP problem with an objective of maximizing the total system throughput.
Constrained MDP formulation has also been presented, where maximizing the total system throughput
is subject to a constraint on the blocking probability of voice users. 
Threshold structures on the association of voice and data users have been derived. Based on the structure of the optimal 
policies, we have proposed two algorithms for the association and offloading of voice/data users in an LTE-WiFi HetNet.
Simulation results demonstrate 
that although the voice user blocking probability performance of Algorithm \ref{algo3} is worse than that of on-the-spot offloading \cite{m}, 
Algorithm \ref{algo4} performs as good as
on-the-spot offloading. Moreover, the proposed algorithms perform better than on-the-spot offloading algorithm in 
improving the total system throughput.
In future, the considered framework can be extended to consider the channel state between LTE BS/WiFi AP and users such that channel-aware 
association and offloading decisions can be taken.
\appendices
\section{Proof of Lemma \ref{lemma1}}\label{app:a}
\begin{figure}[!htb]
 \begin{center}
\scalebox{0.30}{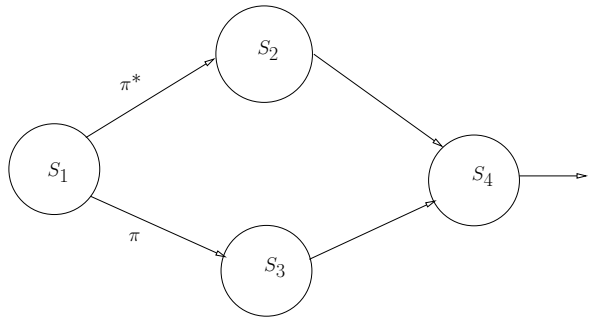}
\caption{Sample path under different policies.}
\label{fig:samplepath}
 \end{center}
\end{figure}
Since the decisions of association and offloading are involved during the arrival and departure of users, proving this lemma
is equivalent to proving the following statements.
\begin{itemize}
 \item[(a)] $A_3$ (Accept in WiFi) is optimal when there are less than $k_{th}$ data users in the system, and a data user arrives.
 \item [(b)] $A_1$ (Do nothing) is optimal when there are less than or equal to $k_{th}$ data users in the system, and a voice user from LTE departs.
 \item [(c)] $A_1$ (Do nothing) is optimal when there are less than or equal to $k_{th}$ data users in the system, and a data user from WiFi departs.
\end{itemize}
We prove the above statements by sample path argument. Suppose the system starts at time $t=0$.\par
\textit{Proof of (a)}:
We consider the scenario when the system is in the state $S_1=(i,0,0)$, when a data user arrival occurs (after a time $t_1$, say).
Assume that the optimal policy $\pi^*$ does not associate this incoming data user with WiFi. Therefore, the optimal action must be $A_2$ (Accept in LTE). As the optimal policy is $\pi^*$,
we have $V^{\pi^*}(s) \ge V^{\hat{\pi}}(s)$, $\forall \hat{\pi}\in \prod$ and $\forall s\in \mathcal{S}$, where $\prod$ is the set of all policies. Let us consider another policy $\pi$ (non-stationary in general) which takes 
$A_3$ in state $S_1=(i,0,0)$. As illustrated in Fig. \ref{fig:samplepath}, let us assume that starting from the state $S_1$ and following the policy $\pi^*$ and $\pi$, the system reaches 
the state $S_2=(i,1,0)$ and $S_3=(i,0,1)$, respectively. The inter-arrival times and service times are same for both the sample paths as we have considered a Markovian system. 
Assume that from the state $S_2$, based on the next event (after a time $t_2$) and the chosen action, the system makes a transition to the state $S_4$ according to the 
policy $\pi^*$. Before reaching the state $S_2$, the sample path followed by the policy $\pi^*$ has one less WiFi data user and one more LTE data user than 
that of the
policy $\pi$ before it reaches the state $S_3$. Suppose, the policy $\pi$ is such that in state $S_3$, it takes the same action as that of policy $\pi^*$ and 
additionally offloads one data user from WiFi to LTE. Evidently, sample path followed by both the policies end up in the same state $S_4$. We construct $\pi$ in such a manner 
that from the state $S_4$ onwards, both the policies take up same actions and follow the same sample path.
Therefore, the difference of value functions of the state $S_1$ under the policy $\pi^*$ and $\pi$ is 
\begin{equation*}
V^{\pi^*}(S_1)-V^{\pi}(S_1)= R_{L,D}-R_{W,D}(1). 
\end{equation*}
Since $R_{L,D}< \tilde {R}_{W,D}(k), \forall k < k_{th}$ and $\tilde {R}_{W,D}(1)=R_{W,D}(1)$, we have, $V^{\pi^*}(S_1) < V^{\pi}(S_1)$.
Clearly, this contradicts the original claim that $\pi^*$ is an optimal policy. 
Since the Markov chains induced by different policies are 
recurrent, the state $(i,0,0)$ is visited infinitely often and each time choice of $A_3$ upon a data user arrival provides more reward than action $A_2$. 
Therefore, when 
there is no data user in the system, and one data user arrives, $A_3$ is optimal. In a similar manner,
it can be proved that $A_3$ is optimal when a data user arrives and the system is in state $(i,0,k)$, where $k< k_{th}$.\par
\textit{Proof of (b) and (c)}:These can be proved using a similar sample path argument.



\section{Proof of Lemma \ref{lemma2}}\label{app:b}
Similar to Lemma \ref{lemma1}, proving this lemma
is equivalent to proving the following statements.
\begin{itemize}
 \item[(a)] $A_2$ (Accept in LTE) is optimal when there are more than or equal to $k_{th}$ data users in the system, and one data user arrives.
 \item[(b)] $A_1$ (Do nothing) is optimal when there are more than $k_{th}$ data users in the system, and a voice/data user from LTE departs.
 \item[(c)] $A_5$ (Data offload to a RAT from where a user 
has departed) is optimal when there are more than $k_{th}$ data users in the system, and a data user from WiFi departs.
\end{itemize}
\textit{Proof of (a)}:
From Lemma \ref{lemma1}, we have, $(j+k)\le k_{th}\implies j=0$.
We consider the scenario when the system is in the state $(i,0,k_{th})$, when a data user arrival occurs.
Assume that the optimal policy $\pi^*$ does not associate this incoming data user with LTE. Consequently, the optimal action must be $A_3$. 
As the optimal policy is $\pi^*$,
we have $V^{\pi^*}(s) \ge V^{\hat{\pi}}(s)$ $\forall \hat{\pi}\in \prod$ in every state $s$. Let us consider another policy $\pi$ which chooses 
$A_2$ in state $(i,0,k_{th})$. As illustrated in Fig. \ref{fig:samplepath}, starting from the state $(i,0,{k_{th}})$ and following the policy $\pi^*$ and $\pi$, 
the system reaches the states $S_2$ and $S_3$, respectively. From the state $S_2$, based on an event, the system reaches the state $S_4$. 
Suppose, in the state $S_3$, the action followed by policy $\pi$ is such that it chooses the same action as that of policy $\pi^*$ and additionally
offloads one data user from LTE to WiFi. Clearly, path followed by both the policies end up in the same state $S_4$. We construct $\pi$ in such a way that
from the state $S_4$ onwards, both of them follow the same path.
Similar to the previous lemma, 
the difference of value functions under the policy $\pi^*$ and $\pi$ is 
\begin{equation*}
V^{\pi^*}(S_1)-V^{\pi}(S_1)=\big((k+1)R_{W,D}(k+1)-kR_{W,D}(k)-R_{L,D}\big).
\end{equation*}

Since $R_{L,D}\ge \tilde {R}_{W,D}(k), \forall k \ge k_{th}$, we have, $V^{\pi^*}(S_1) < V^{\pi}(S_1)$.
Clearly, this contradicts the original claim that $\pi^*$ is an optimal policy. 
Thus, $A_2$ is optimal when there are $k_{th}$ data 
users in WiFi, and one data user arrives. The same result can be extended for the case 
when there are $k_{th}$ data users in WiFi, more than or equal to one data user in LTE, and one data user arrives.\par
Statements (b) and (c) can be proved in a similar way.
\section{Proof of Lemma \ref{lemma5}}\label{app:c}
To prove this lemma, we consider two cases, $(1)$ $k\ge k_{th}$ and $(2)$ $k<k_{th}$. We prove the required for the first case. 
Proof of the second case follows in a similar manner.
From Lemma \ref{lemma4}, we know that for $k\ge k_{th}$, $A_2$ is better than $A_4$. Thus, for $k\ge k_{th}$, the choice is between $A_1$ 
and $A_2$. To prove this lemma, we first prove that the value function $V(i,j,k)$ is concave in $i$. In Lemma \ref{lemma1} and \ref{lemma2}, we have already derived 
the structure of the optimal
policy for data user arrival and departure of voice and data users. Now, for $k\ge k_{th}$, with the aid of this, the optimality equation 
is as follows.
\begin{equation}\label{opt3}
\begin{split}
& V(i,j,k)=\lambda_v\delta\max\lbrace f(i,j,k)-\beta+V(i,j,k),f(i+1,j,k)+V(i+1,j,k)\rbrace
+\lambda_d\delta\big(f(i,j+1,k)+V(i,j+1,k)\big)\\&
+i\mu_v\delta\big(f(i-1,j+1,k-1)+V(i-1,j+1,k-1)\big)
+j\mu_d\delta\big(f(i,j,k-1)+V(i,j,k-1)\big)\\&
+k\mu_d\delta\big(f(i,j,k-1)+V(i,j,k-1)\big)
+\big(1-v(i,j,k)\big)V(i,j,k).
 \end{split}
\end{equation}
Let the components in Equation (\ref{opt3}) be denoted by $V^1(i,j,k),V^2(i,j,k),V^3(i,j,k),
V^4(i,j,k),V^5(i,j,k)$ and $V^6(i,j,k)$, respectively.
We prove the concavity of $V(i,j,k)$ component-wise. 
Start the VIA with $V_0(i,j,k)=0$. Hence, $V_0(i,j,k)$ is concave in $i$.
Let us assume that $V_{1,n}(i,j,k)=\max\lbrace f(i,j,k)-\beta+V_{n-1}(i,j,k),f(i+1,j,k)+V_{n-1}(i+1,j,k)\rbrace$.
Equivalently, $V_{1,n}(i,j,k)=\max\lbrace -\beta+V_{n-1}(i,j,k),R_{L,V}+V_{n-1}(i+1,j,k)\rbrace$.
Let us define the function $V_{1,n}(i,j,k,a)$ as follows.
\begin{equation*}
V_{1,n}(i,j,k,a)=
 \begin{cases}
     -\beta+V_{n-1}(i,j,k),&a=A_1,\\
      R_{L,V}+V_{n-1}(i+1,j,k),&a=A_2.\\
\end{cases}
\end{equation*}
By definition,
\begin{equation*}
 V_{1,n}(i,j,k)=\max_{a\in\lbrace A_1,A_2\rbrace} V_{1,n}(i,j,k,a).
\end{equation*}
Thus, we have, 
\begin{equation*}
 V^1(i,j,k)=\lim_{n\to\infty} V_{1,n}(i,j,k).
\end{equation*}
Let us define
$D_iV(i,j,k,a)=V(i+1,j,k,a)-V(i,j,k,a)$.
\begin{equation*}
D_i V_{1,n}(i,j,k,a)=
 \begin{cases}
        D_i V_{n-1}(i,j,k),&a=A_1,\\
        D_i V_{n-1}(i+1,j,k),&a=A_2.\\
\end{cases}
\end{equation*}
\begin{equation*}
D_{ii} V_{1,n}(i,j,k,a)=
 \begin{cases}
  D_{ii} V_{n-1}(i,j,k),&a=A_1,\\
     D_{ii} V_{n-1}(i+1,j,k),&a=A_2.\\
\end{cases}
\end{equation*}
Since $V_{n-1}(i,j,k)$ is concave in $i$, $V_{1,n}(i,j,k,a)$ is concave in $i$.\par
Now, we need to prove that $V_{1,n}(i,j,k)$ is concave in $i$.
In other words, we need to prove that 
$V_{1,n}(i+2,j,k)+V_{1,n}(i,j,k)\le 2V_{1,n}(i+1,j,k)$.
Let us assume that $a_1\in \lbrace A_1,A_2 \rbrace$ and $a_2\in \lbrace A_1,A_2 \rbrace$ are the maximizing actions in states $(i+2,j,k)$ and $(i,j,k)$, 
respectively.
Therefore,
\begin{equation*}
\begin{split}
&2V_{1,n}(i+1,j,k) \ge V_{1,n}(i+1,j,k,a_1)+V_{1,n}(i+1,j,k,a_2)\\&
=V_{1,n}(i+2,j,k,a_1)+V_{1,n}(i,j,k,a_2)
-D_i V_{1,n}(i+1,j,k,a_1)+D_i V_{1,n}(i,j,k,a_2).
\end{split}
\end{equation*}
Let us take $X=D_i V_{1,n}(i,j,k,a_2)-D_i V_{1,n}(i+1,j,k,a_1)$. To prove that $V_{1,n}(i,j,k)$ is concave in $i$, we need to prove that $X\ge 0$.
There are four cases as described below.\\
$\text{Case } 1: a_1=a_2=A_1,$\\
\begin{equation*}
\begin{split}
& X=D_i V_{n-1}(i,j,k)-D_i V_{n-1}(i+1,j,k)\\&=-D_{ii} V_{n-1}(i,j,k) \ge 0.
\end{split}
\end{equation*}
$\text{Case } 2: a_1=A_1,a_2=A_2,$\\
\begin{equation*}
\begin{split}
 X=D_i V_{n-1}(i+1,j,k)-D_i V_{n-1}(i+1,j,k)= 0.
 \end{split}
\end{equation*}
$\text{Case } 3: a_2=a_2=A_2,$\\
\begin{equation*}
\begin{split}
& X=D_i V_{n-1}(i+1,j,k)-D_i V_{n-1}(i+2,j,k)\\&=-D_{ii} V_{n-1}(i+1,j,k) \ge 0.
 \end{split}
\end{equation*}
$\text{Case } 4: a_1=A_2,a_2=A_1,$\\
\begin{equation*}
\begin{split}
& X=D_i V_{n-1}(i,j,k)-D_i V_{n-1}(i+2,j,k)\\&=-D_{ii} V_{n-1}(i,j,k) -D_{ii} V_{n-1}(i+1,j,k) \ge 0.
\end{split}
\end{equation*}
Thus, it is proved that $V_{1,n}(i,j,k)$ is concave in $i$.
Since this holds for every $n$ and every $\beta$, $V^1(i,j,k)$ is concave in $i$.

Similarly, in the case of the second component, let $V_{2,n}(i,j,k)=f(i+1,j,k)+V_{n-1}(i,j+1,k)$.

Thus, $D_{ii} V_{2,n}(i,j,k)= D_{ii} V_{n-1}(i,j+1,k).$
Therefore, $V_{2,n}(i,j,k)$ is concave in $i$.
Similarly, other components also can be proved to be concave in $i$. Therefore, $V(i,j,k)$ is concave in $i$.\par
Let us define $x(i,j,k)=-\beta-R_{L,V}$.
In order to prove this lemma, we know that if state $(i,j,k)$ is blocking, then $V(i+1,j,k)-V(i,j,k) \le x(i,j,k).$
Due to concavity of $V(i,j,k)$, $V(i+2,j,k)-V(i+1,j,k)\le V(i+1,j,k)-V(i,j,k)$. Now, $x(i,j,k)=x(i+1,j,k)$. 
As a consequence, $V(i+2,j,k)-V(i+1,j,k) \le x(i+1,j,k).$
Thus, it is proved that if state $(i,j,k)$ is blocking, then the state $(i+1,j,k)$ is also blocking.\par
To prove that if state $(i,j,k)$ is blocking, then the state $(i,j+1,k)$ is also blocking, we first need to prove that the value function is 
submodular in $(i,j)$. In other words, we need to prove that
$V_n(i+1,j,k)+V_n(i,j+1,k)\ge V_n(i,j,k)+V_n(i+1,j+1,k)$. Similar to the previous proof, we prove the above statement component-wise. 
Let us assume that $a_1$ and $a_2$ are the maximizing actions in states $(i,j,k)$ and $(i+1,j+1,k)$, respectively.
Start the VIA with $V_0(i,j,k)=0$. Therefore, $V_0(i,j,k)$ is submodular in $(i,j)$. In other words, $D_{ij}V_0(i,j,k)\le 0$.
We have,
\begin{equation*}
 \begin{split}
&V_{1,n}(i+1,j,k)+V_{1,n}(i,j+1,k)\ge V_{1,n}(i+1,j,k,a_1)+V_{1,n}(i,j+1,k,a_2)\\&
=V_{1,n}(i,j,k,a_1)+V_{1,n}(i+1,j+1,k,a_2)+D_i V_{1,n}(i,j,k,a_1)-D_i V_{1,n}(i,j+1,k,a_2).
 \end{split}
\end{equation*}
Now, we consider four possible cases.\\
$\text{Case } 1: a_1=a_2=A_1,$\\
\begin{equation*}
\begin{split}
& D_i V_{1,n}(i,j,k,a_1)-D_i V_{1,n}(i,j+1,k,a_2)\\&=D_i V_{n-1}(i,j,k)-D_i V_{n-1}(i,j+1,k)\\&=-D_{ij} V_{n-1}(i,j,k)\ge 0.
 \end{split}
\end{equation*}
$\text{Case } 2:a_1=a_2=A_2,$\\
\begin{equation*}
\begin{split}
& D_i V_{1,n}(i,j,k,a_1)-D_i V_{1,n}(i,j+1,k,a_2)\\&=D_i V_{n-1}(i+1,j,k)-D_i V_{n-1}(i+1,j+1,k)\\&=-D_{ij} V_{n-1}(i+1,j,k)\ge 0.
 \end{split}
\end{equation*}
$\text{Case } 3: a_1=A_1, a_2=A_2,$\\
\begin{equation*}
\begin{split}
& D_i V_{1,n}(i,j,k,a_1)-D_i V_{1,n}(i,j+1,k,a_2)\\&=D_i V_{n-1}(i,j,k)-D_i V_{n-1}(i+1,j+1,k)\\&=D_i V_{n-1}(i,j,k)-D_i V_{n-1}(i,j+1,k)\\&+D_i V_{n-1}(i,j+1,k)
-D_i V_{n-1}(i+1,j+1,k)\\&=-D_{ij} V_{n-1}(i,j,k)-D_{ii} V_{n-1}(i,j+1,k)\ge 0.
 \end{split}
\end{equation*}
$\text{Case } 4: a_1=A_2, a_2=A_1,$\\
\begin{equation*}
\begin{split}
&V_{1,n}(i+1,j,k)+V_{1,n}(i,j+1,k)\\&\ge V_{1,n}(i+1,j,k,a_2)+V_{1,n}(i,j+1,k,a_1)\\&
=-\beta+V_{n-1}(i+1,j,k)+R_{L,V}+V_{n-1}(i+1,j+1,k)\\&
=V_{1,n}(i,j,k,2)+V_{1,n}(i+1,j+1,k,1)\\&
=V_{1,n}(i,j,k)+V_{1,n}(i+1,j+1,k).
 \end{split}
\end{equation*}
Thus, it is proved that $V_{1,n}(i,j,k)$ is submodular in $(i,j)$.\par
Similarly, in the case of the second component, we have, $V_{2,n}(i,j,k)=f(i,j+1,k)+V_n(i,j+1,k)$.
Therefore, we have, $D_{ij}V_{2,n}(i,j,k)=D_{ij}V_n(i,j+1,k)\le 0$.
Similarly, other components also can be proved to be submodular in $(i,j)$.
Therefore, the value function is submodular in $(i,j)$.\par
Now, if state $(i,j,k)$ is blocking then we have, $V(i+1,j,k)-V(i,j,k) \le x(i,j,k).$ Again, we have,
$x(i,j,k)=x(i,j+1,k)$. Due to submodularity, we have $V(i+1,j+1,k)-V(i,j+1,k)\le V(i+1,j,k)-V(i,j,k) \le x(i,j,k)= x(i,j+1,k)$.
Thus, in the case of voice arrival, if $A_1$ is optimal in state $(i,j,k)$, then in state $(i,j+1,k)$ also $A_1$ is optimal.\par
Proof of $(ii)$ follows directly from the proof of part $(i)$.
\section{Proof of Lemma \ref{lemma6}}\label{app:d}
To prove this lemma, we consider two cases, $(1)$ $k\ge k_{th}$ and $(2)$ $k<k_{th}$. We demonstrate the proof of the lemma for the first case. 
Proof of the second case follows in similar manner.
To prove this lemma, we first need to prove that for $(i+j)=C$, the difference of value functions $V(i+1,j-1,k+1)-V(i,j,k)$ is decreasing in $i$. 
For $(i+j)=C$ and $k \ge k_{th}$, 
the optimality equation 
can be described as
\begin{equation*}
\begin{split}
& V(i,j,k)=\lambda_v\delta\max\lbrace f(i,j,k)-\beta+V(i,j,k),f(i+1,j-1,k+1)+V(i+1,j-1,k+1)\rbrace\\&
+\lambda_d\delta\big( f(i,j,k+1)+V(i,j,k+1)\big)
+i\mu_v\delta\big(f(i-1,j+1,k-1)+V(i-1,j+1,k-1)\big)\\&
+j\mu_d\delta\big(f(i,j,k-1)+V(i,j,k-1)\big)
+k\mu_d\delta\big(f(i,j,k-1)+V(i,j,k-1)\big)\\&
+\big(1-v(i,j,k)\big)V(i,j,k).
 \end{split}
\end{equation*}

Let us assume that $V_{1,n}(i,j,k,a)=\max\lbrace f(i,j,k)-\beta+V_{n-1}(i,j,k),f(i+1,j-1,k+1)+V_{n-1}(i+1,j-1,k+1)\rbrace$.\\
Equivalently, $V_{1,n}(i,j,k,a)=\max\lbrace-\beta+ V_{n-1}(i,j,k),\\R_{L,V}-R_{L,D}+\tilde{R}_{W,D}(k)+V_{n-1}(i+1,j-1,k+1)\rbrace$.

In other words, 
\begin{equation*}
V_{1,n}(i,j,k,a)=
 \begin{cases}
   -\beta +V_{n-1}(i,j,k),&a=A_1,\\
     R_{L,V}-R_{L,D}+\tilde{R}_{W,D}(k)+\\V_{n-1}(i+1,j-1,k+1),&a=A_4.\\
\end{cases}
\end{equation*}
We prove the above claim component-wise. 
Start the VIA with $V_0(i,j,k)=0$. Therefore, $V_0(i+1,j-1,k+1)-V_0(i,j,k)$ is decreasing in $i$. We also have, $E_iF_iV_0(i,j,k)\le 0$.
Now,
Let us define
$E_iV(i,j,k,a)=V(i+1,j-1,k+1,a)-V(i,j,k,a)$ and $F_iV(i,j,k,a)=V(i+1,j-1,k,a)-V(i,j,k,a)$.

\begin{equation*}
E_i V_{1,n}(i,j,k,a)=
 \begin{cases}
   E_i V_{n-1}(i,j,k),&a=A_1,\\
   E_i V_{n-1}(i+1,j-1,k+1),&a=A_4.\\
\end{cases}
\end{equation*}
Therefore, $E_iF_i V_{1,n}(i,j,k,a)\le 0$.\\
\begin{equation*}
E_{ii} V_{1,n}(i,j,k,a)=
 \begin{cases}
   E_{ii} V_{n-1}(i,j,k),&a=A_1,\\
     E_{ii} V_{n-1}(i+1,j-1,k+1),&a=A_4.\\
\end{cases}
\end{equation*}
Therefore, $V_{1,n}(i+1,j-1,k+1,a)-V_{1,n}(i,j,k,a)$ is decreasing in $i$.\\
Now, we need to prove that $V_{1,n}(i+1,j-1,k+1)-V_{1,n}(i,j,k)$ is decreasing in $i$.
In other words, we need to prove that 
$V_{1,n}(i+2,j-2,k+1)+V_{1,n}(i,j,k)\le V_{1,n}(i+1,j-1,k+1)+V_{1,n}(i+1,j-1,k)$.
Let us assume that $a_1\in \lbrace A_1,A_4 \rbrace$ and $a_2\in \lbrace A_1,A_4 \rbrace$ are the maximizing actions in states $(i+2,j-2,k+1)$ and $(i,j,k)$, 
respectively.
Therefore,
\begin{equation*}
\begin{split}
&V_{1,n}(i+1,j-1,k+1)+V_{1,n}(i+1,j-1,k)\\&\ge V_{1,n}(i+1,j-1,k+1,a_2)+V_{1,n}(i+1,j-1,k,a_1)\\&
=V_{1,n}(i,j,k,a_2)+V_{1,n}(i+2,j-2,k+1,a_1)\\&
-E_i V_{1,n}(i+1,j-1,k,a_1)+E_i V_{1,n}(i,j,k,a_2).
\end{split}
\end{equation*}
Let us take $Y=E_i V_{1,n}(i,j,k,a_2)-E_i V_{1,n}(i+1,j-1,k,a_1)$. To prove that $V_{1,n}(i+1,j-1,k+1)-V_{1,n}(i,j,k)$ is decreasing in $i$, we need to 
prove that $Y\ge 0$.
There are four cases as described below.\\
$\text{Case } 1: a_1=a_2=A_1,$\\
\begin{equation*}
\begin{split}
& Y=E_i V_{n-1}(i,j,k)-E_i V_{n-1}(i+1,j-1,k)\\&=-E_i F_i V_{n-1}(i,j,k) \ge 0.
\end{split}
\end{equation*}
$\text{Case } 2: a_1=A_4,a_2=A_1,$\\
\begin{equation*}
\begin{split}
 &Y=E_i V_{n-1}(i,j,k)-E_i V_{n-1}(i+2,j-2,k+1)\\&= -E_{ii}V_{n-1}(i,j,k)-E_iF_i(i+1,j-1,k+1)\ge 0.
 \end{split}
\end{equation*}
$\text{Case } 3: a_2=a_2=A_4,$\\
\begin{equation*}
\begin{split}
& Y=E_i V_{n-1}(i+1,j-1,k+1)-E_i V_{n-1}(i+2,j-2,k+1)\\&=-E_iF_i V_{n-1}(i+1,j-1,k+1) \ge 0.
 \end{split}
\end{equation*}
$\text{Case } 4: a_1=A_1,a_2=A_4,$\\
\begin{equation*}
\begin{split}
&V_{1,n}(i+1,j-1,k+1)+V_{1,n}(i+1,j-1,k)\\&\ge V_{1,n}(i+1,j-1,k+1,a_1)+V_{1,n}(i+1,j-1,k,a_2)\\&
=-\beta+V_{n-1}(i+1,j-1,k+1)+R_{L,V}-R_{L,D}+\\&\tilde{R}_{W,D}(k)+V_{n-1}(i+2,j-2,k+1)\\&
=V_{1,n}(i+2,j-2,k+1,1)+V_{1,n}(i,j,k,4)\\&
=V_{1,n}(i+2,j-2,k+1)+V_{1,n}(i,j,k).
\end{split}
\end{equation*}
Thus, it is proved that $V_{1,n}(i+1,j-1,k+1)-V_{1,n}(i,j,k)$ is decreasing in $i$.
Since this holds for every $n$ and every value of $\beta$, $V^1(i+1,j-1,k+1)-V^1(i,j,k)$ is decreasing in $i$.

Let $V_{2,n}(i,j,k)=f(i,j,k+1)+V_{n-1}(i,j,k+1)$.
Thus, $E_{ii} V_{2,n}(i,j,k)= E_{ii} V_{n-1}(i,j,k+1).$
Therefore, $V_{2,n}(i+1,j-1,k+1)-V_{2,n}(i,j,k)$ is decreasing in $i$.
Similarly, other components can be proved to be decreasing in $i$. Therefore, $V(i+1,j-1,k+1)-V(i,j,k)$ is decreasing in $i$.\par
Similar to Lemma \ref{lemma5}, using this property it can be proved that if the optimal action for voice user arrival in state $(i,j,k)$ is
blocking, then the optimal action in state $(i + 1,j-1,k)$
is also blocking.\par
Proof of $(ii)$ follows directly from the proof of part $(i)$.



%


\section*{Acknowledgment}
This work is funded by the Department of Electronics and Information Technology (DeitY), Government of India.


\ifCLASSOPTIONcaptionsoff
  \newpage
\fi

\end{document}